\documentclass[a4paper]{article}
\usepackage[utf8]{inputenc}

\usepackage{tikz}
\usepackage{graphicx}
\usepackage{wrapfig}
\usepackage{caption}
\usepackage{subcaption}
\usepackage[ruled,vlined]{algorithm2e}
\SetKw{Continue}{continue}

\usepackage{fullpage}

\usepackage{amsmath}
\usepackage{amsfonts}
\usepackage{amsthm}
\newtheorem{theorem}{Theorem}[section]
\newtheorem{corollary}{Corollary}[theorem]
\newtheorem{lemma}[theorem]{Lemma}

\usepackage{lipsum}

\title{Folding Every Point on a Polygon Boundary to a Point}

\author{
  Nattawut Phetmak\thanks{Email: nattawut.p@ku.th. Department of Computer Engineering, Kasetsart University, Bangkok, Thailand.}
  \and
  Jittat Fakcharoenphol\thanks{Email: jittat@gmail.com. Department of Computer Engineering, Kasetsart University, Bangkok, Thailand. Supported by the Thailand Research Fund, Grant RSA-6180074.}
}

\date{}

\begin{document}

\maketitle

\begin{abstract}
  We consider a problem in computational origami. 
Given a piece of paper as a convex polygon $P$ and a point $f$ located within,
fold every point on a boundary of $P$ to $f$ and compute a region that is safe from folding, i.e., the region with no creases.
This problem is an extended version of a problem by Akitaya, Ballinger, Demaine, Hull, and Schmidt~[CCCG'21] that only folds corners of the polygon.
To find the region, we prove structural properties of intersections of parabola-bounded regions
and use them to devise a linear-time algorithm.
We also prove a structural result regarding the complexity of the safe region as a variable of the location of point $f$,
i.e., the number of arcs of the safe region can be determined using the straight skeleton of the polygon $P$.
\end{abstract}


\section{Introduction}
\label{sect:intro}

Paper folding offers rich computational geometry problems with many real-world applications~\cite{Lang-comp-origami-2009}.
The topic, typically referred to as {\em computational origami} or {\em mathematics of paper folding}~\cite{DO08book,hull2020origametry},
studies both feasibility problems and also structural problems~\cite{AkitayaBDH021, AkitayaDK17-simple-folding-hard,AbelDDELU-flat-folding}
with the aim to illuminate the connections between physical structures/problems and mathematical geometric objects (see, e.g.,~\cite{Dambrogio-unlocking-21, FeltonTDRW-2014-self-folding}).  
As geometric construction using straightedge and compass offers elegant connections between algebra and geometry,
paper folding, which can be seen as geometric construction with additional operations, may provide beautiful structural properties worth studying
(e.g., see \cite{DO08book,hull2011solving}).

\begin{figure}
  \centering
  \begin{subfigure}[t]{0.45\textwidth}
    \centering
    \begin{tikzpicture}
      \draw[semithick, olive]  (0.989,2.744) -- (0.733,1.678);
      \draw[semithick, red]    (0.851,2.638) -- (1.230,1.278);
      \draw[semithick, teal]   (0.913,1.532) -- (5.369,1.169);
      \draw[semithick, orange] (1.423,1.121) -- (4.123,2.401);
      \draw[semithick, violet] (1.764,0.847) -- (4.039,2.442);
      \draw[stealth-stealth, densely dashed, olive]  (0.180,2.124) .. controls (0.592,2.126) and (1.139,2.069) .. (1.402,1.864);
      \draw[stealth-stealth, densely dashed, red]    (0.776,1.643) .. controls (0.893,1.793) and (1.126,1.843) .. (1.364,1.820);
      \draw[stealth-stealth, densely dashed, teal]   (1.374,1.162) -- (1.408,1.770);
      \draw[stealth-stealth, densely dashed, orange] (1.972,0.679) .. controls (1.955,0.991) and (1.742,1.540) .. (1.474,1.778);
      \draw[stealth-stealth, densely dashed, violet] (2.568,0.198) .. controls (2.560,0.658) and (1.919,1.752) .. (1.493,1.818);
      \draw[thick] (1.932,3.468) -- (0.180,2.124) -- (2.568,0.198) -- (4.734,0.534) -- (5.790,1.590) -- cycle;
      \fill (1.428,1.824) circle (1.5pt) node[anchor=south west] {$\scriptstyle f$};
      \node at (3.7,2.9) {$P$};
      \node[anchor=west]  at (1.932,3.468) {$\scriptstyle v_1$};
      \node[anchor=north] at (0.180,2.124) {$\scriptstyle v_2$};
      \node[anchor=east]  at (2.568,0.198) {$\scriptstyle v_3$};
      \node[anchor=west]  at (4.734,0.534) {$\scriptstyle v_4$};
      \node[anchor=south] at (5.790,1.590) {$\scriptstyle v_5$};
    \end{tikzpicture}
    \caption{Sample creases from one edge $\overline{v_2v_3}$}
  \end{subfigure}
  \hfill
  \begin{subfigure}[t]{0.45\textwidth}
    \centering
    \begin{tikzpicture}
      \fill[blue!20]
        (2.035,2.449) .. controls (1.586,2.476) and (1.244,2.364) ..
        (1.007,2.114) .. controls (1.108,1.420) and (1.663,1.288) ..
        (2.672,1.718) .. controls (2.495,2.033) and (2.281,2.275) .. cycle;
      \node[blue] at (2.3,1.9) {$R$};
      \draw[densely dotted, black] (5.634,1.434) -- (1.094,2.826); 
      \draw[densely dotted, black] (4.876,2.035) -- (0.848,2.636);
      \draw[densely dotted, black] (3.756,2.580) -- (0.490,2.362);
      \draw[densely dotted, black] (2.960,2.968) -- (0.365,1.974);
      \draw[densely dotted, black] (2.316,3.281) -- (0.617,1.772);
      \draw[densely dotted, black] (1.544,3.170) -- (0.728,1.681);
      \draw[densely dotted, black] (1.010,2.760) -- (0.932,1.517);
      \draw[densely dotted, black] (0.420,2.308) -- (2.887,0.247);
      \draw[densely dotted, black] (1.223,1.283) -- (4.564,2.188);
      \draw[densely dotted, black] (1.598,0.979) -- (4.021,2.452);
      \draw[densely dotted, black] (1.879,0.754) -- (3.779,2.569);
      \draw[densely dotted, black] (2.015,0.644) -- (3.636,2.639);
      \draw[densely dotted, black] (2.165,0.523) -- (3.572,2.670);
      \draw[densely dotted, black] (2.324,0.395) -- (3.562,2.675);
      \draw[densely dotted, black] (2.492,0.259) -- (3.590,2.660);
      \draw[densely dotted, black] (2.707,0.220) -- (3.648,2.633); 
      \draw[densely dotted, black] (2.909,0.251) -- (3.600,2.657);
      \draw[densely dotted, black] (3.114,0.283) -- (3.587,2.663);
      \draw[densely dotted, black] (3.324,0.316) -- (3.608,2.652);
      \draw[densely dotted, black] (3.536,0.348) -- (3.659,2.628); 
      \draw[densely dotted, black] (3.398,0.326) -- (3.427,2.741);
      \draw[densely dotted, black] (3.286,0.310) -- (3.185,2.858);
      \draw[densely dotted, black] (3.204,0.296) -- (2.928,2.983);
      \draw[densely dotted, black] (3.166,0.290) -- (2.647,3.120);
      \draw[densely dotted, black] (3.190,0.294) -- (2.326,3.276);
      \draw[densely dotted, black] (3.306,0.312) -- (1.930,3.467);
      \draw[densely dotted, black] (3.566,0.353) -- (1.672,3.268);
      \draw[densely dotted, black] (4.080,0.432) -- (1.439,3.089);
      \draw[densely dotted, black] (5.087,0.589) -- (1.241,2.938);
      \draw[semithick, olive]  (0.989,2.744) -- (0.733,1.678);
      \draw[semithick, red]    (0.851,2.638) -- (1.230,1.278);
      \draw[semithick, teal]   (0.913,1.532) -- (5.369,1.169);
      \draw[semithick, orange] (1.423,1.121) -- (4.123,2.401);
      \draw[semithick, violet] (1.764,0.847) -- (4.039,2.442);
      \draw[thick] (1.932,3.468) -- (0.180,2.124) -- (2.568,0.198) -- (4.734,0.534) -- (5.790,1.590) -- cycle;
      \fill (1.428,1.824) circle (1.5pt) node[anchor=south west] {$\scriptstyle f$};
      \node at (3.7,2.9) {$P$};
      \node[anchor=west]  at (1.932,3.468) {$\scriptstyle v_1$};
      \node[anchor=north] at (0.180,2.124) {$\scriptstyle v_2$};
      \node[anchor=east]  at (2.568,0.198) {$\scriptstyle v_3$};
      \node[anchor=west]  at (4.734,0.534) {$\scriptstyle v_4$};
      \node[anchor=south] at (5.790,1.590) {$\scriptstyle v_5$};
    \end{tikzpicture}
    \caption{The safe region $R$}
  \end{subfigure}
  \caption{The problem of folding every point on the boundary of a polygon $P$ to a point $f$}
  \label{fig:hook}
\end{figure}
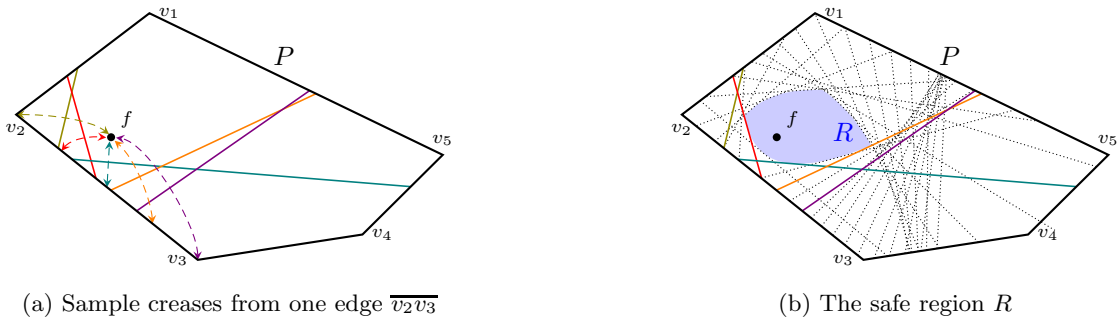

Akitaya, Ballinger, Demaine, Hull, and Schmidt~\cite{AkitayaBDH021} previously considered two folding problems on a convex piece of paper $P$.
Given a query point $p$ inside $P$, their first problem, originally proposed by Haga~\cite{haga1994proposal,haga2008origamics,hull-project-origami}, called {\em points to a point}, is to find a region containing $p$ bounded by creases after fold-and-unfold each corner of $P$ onto $p$.
Their second problem, called {\em lines to a line}, takes as an input a line $\ell$ inside $P$ and the goal is to find a region containing $\ell$ bounded by creases after fold-and-unfold each side of $P$ onto $\ell$.
Although these two problems share a similar structure, they find that the outcomes diverge.
That is, the region in the first problem resembles a Voronoi cell~\cite{voronoi-book-2013} with the inner point and corners as seeds,
while the region in the second problem relies on the straight skeleton of the piece of paper.

We consider a variance of this folding problem in the same flavor, i.e., we are given a query point $f$ inside $P$ and we are interested in a region containing $f$ bounded by creases after fold-and-unfold {\em every} point on the boundary of the paper $\delta P$ onto $f$.
We call this result region a {\em safe region} $R$, since every point $p \in R$ is safe from this folding procedure.  
Our contributions are an analysis of the shape of $R$, an efficient algorithm for finding $R$, and a complexity of the safe region with respect to any query point $f$.

\begin{figure}
  \centering
  \begin{subfigure}[t]{0.45\textwidth}
    \centering
    \begin{tikzpicture}
      \draw[densely dotted, black] (1.094,2.826) -- (5.634,1.434);
      \draw[densely dotted, black] (0.989,2.745) -- (0.733,1.678);
      \draw[densely dotted, black] (1.764,0.847) -- (4.039,2.442);
      \draw[densely dotted, black] (2.707,0.220) -- (3.648,2.633);
      \draw[densely dotted, black] (3.536,0.348) -- (3.658,2.628);
      \filldraw[blue!20, draw=blue, thick] (0.989,2.745) -- (0.733,1.678) -- (1.764,0.847) -- (3.399,1.993) -- (3.443,2.106) -- (1.094,2.826) -- cycle;
      \fill[blue] (0.989,2.745) circle (1pt);
      \fill[blue] (0.733,1.678) circle (1pt);
      \fill[blue] (1.764,0.847) circle (1pt);
      \fill[blue] (3.399,1.993) circle (1pt);
      \fill[blue] (3.443,2.106) circle (1pt);
      \fill[blue] (1.094,2.826) circle (1pt);
      \draw[thick] (1.932,3.468) -- (0.180,2.124) -- (2.568,0.198) -- (4.734,0.534) -- (5.790,1.590) -- cycle;
      \fill (1.428,1.824) circle (1.5pt) node[anchor=south west] {$\scriptstyle f$};
      \node at (3.7,2.9) {$P$};
      \node[anchor=west]  at (1.932,3.468) {$\scriptstyle v_1$};
      \node[anchor=north] at (0.180,2.124) {$\scriptstyle v_2$};
      \node[anchor=east]  at (2.568,0.198) {$\scriptstyle v_3$};
      \node[anchor=west]  at (4.734,0.534) {$\scriptstyle v_4$};
      \node[anchor=south] at (5.790,1.590) {$\scriptstyle v_5$};
    \end{tikzpicture}
    \caption{Corner folding in \cite{AkitayaBDH021}}
  \end{subfigure}
  \hfill
  \begin{subfigure}[t]{0.45\textwidth}
    \centering
    \begin{tikzpicture}
      \draw[densely dotted, black] (0.743,1.670) .. controls (1.181,2.880) and (2.697,2.687) .. (5.290,1.090);
      \draw[densely dotted, black] (1.020,2.768) .. controls (0.800,1.151) and (1.799,1.046) .. (4.018,2.453);
      \draw[densely dotted, black] (1.629,0.955) .. controls (2.341,1.092) and (2.984,1.666) .. (3.560,2.675);
      \draw[densely dotted, black] (2.291,0.421) .. controls (2.950,0.941) and (3.382,1.688) .. (3.587,2.662);
      \draw[densely dotted, black] (1.020,2.768) .. controls (2.125,2.743) and (2.840,1.918) .. (3.165,0.291);
      \filldraw[blue!20, draw=blue, thick]
        (2.035,2.449) .. controls (1.586,2.476) and (1.244,2.364) ..
        (1.007,2.114) .. controls (1.108,1.420) and (1.663,1.288) ..
        (2.672,1.718) .. controls (2.495,2.033) and (2.281,2.275) .. cycle;
      \fill[blue] (2.035,2.449) circle (1pt);
      \fill[blue] (1.007,2.114) circle (1pt);
      \fill[blue] (2.672,1.718) circle (1pt);
      \draw[thick] (1.932,3.468) -- (0.180,2.124) -- (2.568,0.198) -- (4.734,0.534) -- (5.790,1.590) -- cycle;
      \fill (1.428,1.824) circle (1.5pt) node[anchor=south west] {$\scriptstyle f$};
      \node at (3.7,2.9) {$P$};
      \node[anchor=west]  at (1.932,3.468) {$\scriptstyle v_1$};
      \node[anchor=north] at (0.180,2.124) {$\scriptstyle v_2$};
      \node[anchor=east]  at (2.568,0.198) {$\scriptstyle v_3$};
      \node[anchor=west]  at (4.734,0.534) {$\scriptstyle v_4$};
      \node[anchor=south] at (5.790,1.590) {$\scriptstyle v_5$};
    \end{tikzpicture}
    \caption{Boundary folding in our paper}
  \end{subfigure}
  \caption{Comparison of the setting of the problem}
  \label{fig:compare}
\end{figure}
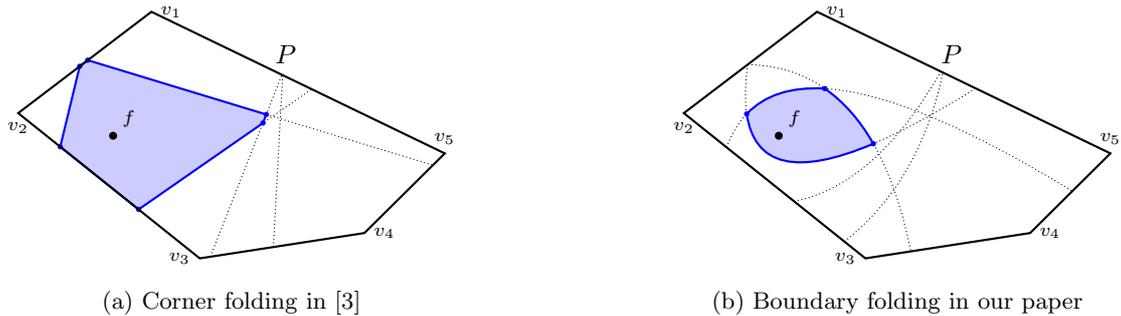

We first show that, although we fold infinitely many points on $\delta P$, the result region $R$ can be described finitely.
It is a well-known result that fold-and-unfold multiple points from a line onto a point produces an envelope of a parabola.
Thus we consider each ``side'' of $R$ as a parabolic arc instead of a traditional straight line segment.
The analysis is given in Section~\ref{sect:prelim}.

Using the analysis on the properties of safe regions, we give a linear-time algorithm for finding the region $R$, given a piece of paper $P$ as a sorted list of $n$ corners in counter-clockwise order.
Our algorithm works similarly to Graham's scan for a convex hull.
That is, we consider adding one side of $P$ as a parabola arc at a time, maintaining the loop invariant that regulates the region $R$.
During each iteration, we may destroy some of the previously added parabola arcs.
The key insight is that the destroyed parabola arc must be the one that is closest to the newly added parabolic region.
Thus, it allows us to amortize the cost of destroying, achieving a linear-time algorithm overall.
Section~\ref{sect:algo} explains the algorithm in full detail.

Finally, given any potential query point $f$, we calculate the precise number of parabola arcs of the region $R$.
They turn out to be dependent on a set of inscribed circles, each of which is centered at a node of the straight skeleton of $P$.
We dedicate Section~\ref{sect:skeleton} to focus solely on this property.

We hope that our problem will expand the richness to the family of fold-and-unfold origami problems.
Nevertheless, it could serve as a {\em bridge} joining the previous two problems from \cite{AkitayaBDH021},
since our problem statement is similar to their first problem, but our result is similar to their second problem.


\section{Preliminaries}
\label{sect:prelim}

In this section, we give a formal definition of the problem.
We are given a polygon $P$ as an ordered list of vertices $V(P)=[v_1,v_2,\dots,v_n]$, counter-clockwisely.
We shall treat the list as a circular list.
Given two points $a$ and $b$, we refer to a line segment whose ends are $a$ and $b$ as a segment $\overline{ab}$.
Equivalently, the polygon $P$ is also represented as a list of edges $E(P)=[e_1,e_2,\dots,e_n]$, where $e_i=\overline{v_iv_{i+1}}$ is a line segment joining $v_i$ and $v_{i+1}$.
We denote its boundary as $\delta P$, i.e., $\delta P = \bigcup E(P)$.

We are also given a point $f$ strictly inside $P$.
Consider a point $u\in\delta P$, folding $u$ onto $f$ results in a straight line $L$,
referred to as a {\em crease} line, which passes through the middle point on the line segment $\overline{uf}$ and is also orthogonal to $\overline{uf}$.
We are interested in the set of crease lines when folding every point on $e_i=\overline{v_iv_{i+1}}$ onto $f$.  
The envelope of this family of crease lines corresponds to a parabola,
whose focus is $f$ and directrix is a line resulting from extending the segment $e_i$ to a line.  
We refer to this parabola as $p_i$.
See Figure~\ref{fig:fold-parabola}.

\begin{figure}
  \centering
    \begin{tikzpicture}
      \clip (-2.5,0) rectangle (2.5,5);
      \draw[densely dotted, black] (-2.500,3.618) -- (-1.791,0.000);
      \draw[densely dotted, black] (-2.500,4.761) -- (-1.184,0.000);
      \draw[densely dotted, black] (-2.500,4.515) -- (-0.676,0.000);
      \draw[densely dotted, black] (0.334,-0.000) -- (-2.500,2.608);
      \draw[densely dotted, black] (1.412,-0.000) -- (-2.500,1.565);
      \draw[densely dotted, black] (2.500,0.625) -- (-2.500,0.625);
      \draw[densely dotted, black] (2.500,1.565) -- (-1.412,0.000);
      \draw[densely dotted, black] (-0.334,0.000) -- (2.500,2.608);
      \draw[densely dotted, black] (0.207,0.000) -- (2.500,3.660);
      \draw[densely dotted, black] (2.500,4.515) -- (0.676,0.000);
      \draw[densely dotted, black] (2.500,4.761) -- (1.184,0.000);
      \draw[densely dotted, black] (2.500,3.618) -- (1.791,0.000);
      \draw[red, thick] (-0.207,0.000) -- (-2.500,3.660);
      \fill (0,1) circle (1pt) node[anchor=south] {$\scriptstyle f$};
      \draw[thick] (-5,1/4) -- (5,1/4) node[pos=0.69, anchor=south] {\scriptsize directrix};
      \draw[red, stealth-stealth, densely dashed] (-1.187,0.260) .. controls (-0.990,0.575) and (-0.448,0.909) .. (-0.040,0.990);
      \draw[red, thick, densely dotted] (-1.197,1.580) -- (0.000,1.000);
      \draw[red, thick, densely dotted] (-1.197,0.250) -- (-1.197,1.580);
      \fill[red] (-1.185,0.260) circle (1pt) node[anchor=north] {$\scriptstyle u$};
      \fill[red] (-1.185,1.580) circle (1pt) node[anchor=south] {$\;\scriptstyle u'$};
    \end{tikzpicture}
  \caption{Fold-and-unfold points on a line to a fixed point multiple time, creating an envelope of a parabola}
  \label{fig:fold-parabola}
\end{figure}
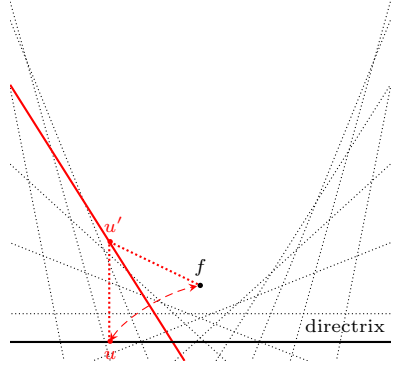

Let line $L'$ erected perpendicularly to the directrix at point $u$, find $u' = L \cap L'$.
We have $|\overline{u'f}| = |\overline{u'u}|$, which indeed obey the definition of parabola.
In other words, $u'$ is a point on parabola's curve.
Furthermore, for $v \in L'$ such that $|\overline{vf}| < |\overline{vu}|$,
the point $v$ is ``safe'' from other creases produced by fold-and-unfold every other points from this directrix.

Since $p_i$ divides the plane into two parabolic half-planes, which we are interested in the region containing $f$.
We define a half-space $H(p_i)$ to be the half-plane containing $f$.
More formally, $H(p_i)$ contains all point $v$ such that $|\overline{vf}| \leq |\overline{vu}|$, where $u$ is an orthogonal projection of $v$ on the line extension of $e_i$.

Our goal is to find
\[
R = \bigcap_{i=1}^n H(p_i),
\]
defined to be the {\em safe region}.
When focusing on point $f$, we sometimes refer to the safe region with respect to point $f$ as $R_f$.

We describe the algorithm in Section~\ref{sect:algo}.
Later in this section, we state relevant geometry facts.

\subsection{A parabola as a projection of a conic section}

A parabola can be viewed as a conic section, i.e., a curve on a surface of a cone intersecting a cutting plane tilted at the same angle of the cone.
Analytically, we consider a Euclidean space $\mathbb{R}^3$.
A {\em cone} is a surface satisfying the equation 
\[
x^2 + y^2 - z^2 = 0
\]
with an apex of the cone at the origin.
A {\em cutting plane} can be defined with an equation 
\[
x\cos\theta + y\sin\theta - z = r,
\]
where $r$ is the distance on $xy$-plane from cone's apex to the nearest point of the cutting plane,
and $\theta$ is the {\em directional angle} on $xy$-plane to that point.
By this definition, we have a parabola as a curve on the tilted cutting plane.
See Figure~\ref{fig:conic-one-parabola} and \ref{fig:conic-projection}.

A projected parabola is an orthogonal projection of the tilted parabola onto $xy$-plane, which is also a parabola.
The projected parabola have cone's apex as its focus, and an intersect line of the cutting plane with $xy$-plane as its directrix.
This projection viewpoint was briefly mention at the end of \cite{fortune1986sweepline}.
In this paper, we refer to projected parabolas simply as parabolas.

We say that parabola $p$ is in the {\em upright form} if, by rotating and translating $xy$-plane,
the parabola possess an analytical form of $y=tx^2$ where $t > 0$.

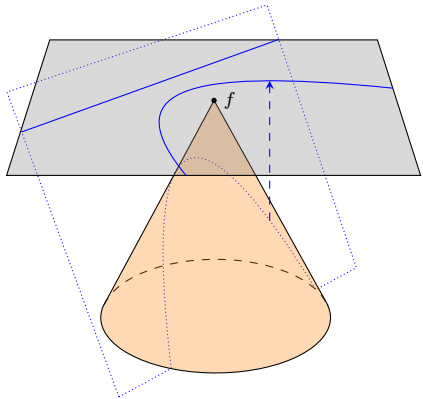
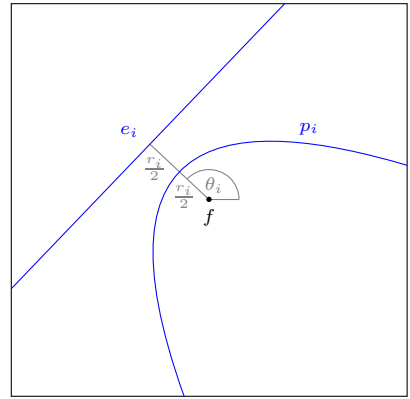
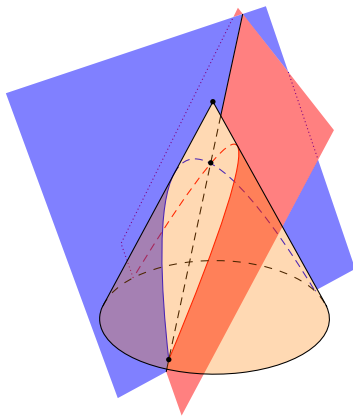
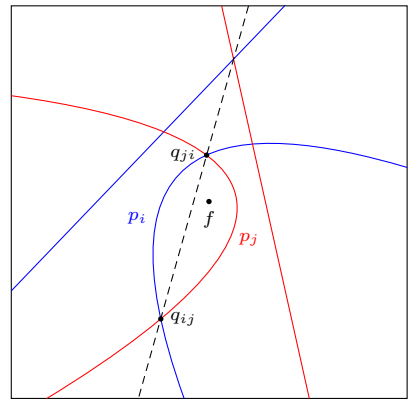
\begin{figure}
  \centering
  \begin{subfigure}[t]{0.45\textwidth}
    \centering
    \begin{tikzpicture}[yscale=-1]
      \draw[blue, densely dotted]
      (1.703,6.794) --
      (2.391,6.417) .. controls (1.974,2.840) and (2.809,2.926) .. (4.294,5.371) --
      (4.826,5.080) --
      (3.651,1.597) --
      (0.231,2.751) --
      cycle;
      \draw[dashed] (1.509,5.571) arc (192.74:347.56:1.503 and 0.771);
      \draw[fill=orange, fill opacity=0.3]
      (4.446,5.574) arc (346.74:553.23:1.511 and 0.734) --
      (1.506,5.574) --
      ++(1.449,-2.711) --
      ++(1.494,2.711) --
      cycle;
      \draw[fill=gray, fill opacity=0.3]
      (0.226,3.857) --
      (0.794,2.060) --
      (5.103,2.060) --
      (5.671,3.857) --
      cycle;
      \draw[blue] (2.586,3.857) .. controls (1.694,2.769) and (2.446,2.406) .. (5.306,2.700);
      \draw[blue] (0.426,3.280) -- (3.806,2.054);
      \fill (2.954,2.863) circle (1pt) node[anchor=west] {$\scriptstyle f$};
      \draw[blue, dashed, -stealth] (3.686,4.454) -- ++(0,-1.851);
    \end{tikzpicture}
    \caption{Mapping a tilted parabola to $xy$-plane}
    \label{fig:conic-one-parabola}
  \end{subfigure}
  \hfill
  \begin{subfigure}[t]{0.45\textwidth}
    \centering
    \begin{tikzpicture}[yscale=-1]
      \draw (7.334,6.900) -- ++(0.000,-5.206) -- ++(5.206,0.000) -- ++(0.000,5.206) -- cycle;
      \draw[gray] (9.649,4.017) arc (223.44:360.29:0.394);
      \draw[gray] (10.331,4.291) -- ++(-0.397,0.000);
      \draw[gray] (9.934,4.291) -- ++(-0.388,-0.367) node[pos=0.85, anchor=north] {$\scriptstyle \frac{r_i}2$};
      \draw[gray] (9.546,3.924) -- ++(-0.388,-0.367) node[pos=0.85, anchor=north] {$\scriptstyle \frac{r_i}2$};
      \node[gray] at (10.0,4.1) {$\scriptstyle\theta_i$};
      \draw[blue] (9.609,6.906) .. controls (8.311,3.260) and (10.306,3.174) .. (12.540,3.837) node[pos=0.8, anchor=south] {$\scriptstyle p_i$};
      \draw[blue] (7.334,5.474) -- (10.934,1.689) node[pos=0.5, anchor=south east] {$\scriptstyle e_i$};
      \fill (9.934,4.289) circle (1pt) node[anchor=north] {$\scriptstyle f$};
    \end{tikzpicture}
    \caption{Projection of a parabola on $xy$-plane}
    \label{fig:conic-projection}
  \end{subfigure}
  \begin{subfigure}[t]{0.45\textwidth}
    \centering
    \begin{tikzpicture}[yscale=-1]
      \draw[blue] (2.374,6.283) .. controls ++(-0.109,-0.974) and ++(-0.226,0.417) .. ++(0.071,-2.469);
      \draw[blue, densely dashed] (2.446,3.814) .. controls ++(0.311,-0.589) and ++(-0.797,-1.320) .. ++(1.851,1.560);
      \draw[blue, densely dotted] (4.297,5.374) -- (4.326,5.354);
      \draw[blue, densely dotted] (3.946,2.469) -- (4.337,3.631);
      \draw[red] (2.374,6.283) .. controls (3.080,4.583) and (3.374,3.629) .. (3.277,3.451);
      \draw[red, densely dashed] (3.277,3.451) .. controls ++(-0.140,-0.220) and ++(0.423,-0.706) .. ++(-1.371,1.749);
      \draw[red, densely dotted] (1.906,5.200) -- (1.746,4.743) -- (3.223,1.740);
      \fill[blue, opacity=0.5]
      (1.703,6.794) --
      ++(0.637,-0.349) --
      ++(0.034,-0.163) .. controls ++(-0.109,-0.974) and ++(-0.226,0.417) .. ++(0.071,-2.469) --
      ++(0.509,-0.951) --
      (3.057,3.051) -- 
      ++(0.286,-1.354) --
      ++(0.603,0.771) --
      (3.651,1.597) --
      (0.231,2.751) --
      cycle;
      \fill[blue, opacity=0.5]
      (3.866,4.523) --
      (4.326,5.354) --
      (4.826,5.080) --
      (4.337,3.631) --
      cycle;
      \fill[red, opacity=0.5]
      (2.543,7.026) --
      ++(-0.203,-0.580) --
      ++(0.034,-0.163) .. controls (3.080,4.583) and (3.374,3.629) .. (3.277,3.451) --
      (3.057,3.051) -- 
      ++(0.289,-1.354) --
      ++(-0.123,0.043) --
      ++(0.060,-0.123) --
      ++(1.263,1.623) --
      cycle;
      \draw (3.346,1.700) -- (3.057,3.051);
      \draw (2.374,6.283) -- (2.340,6.446);
      \draw[densely dashed] (3.057,3.051) -- (2.374,6.283);
      \draw[dashed] (1.509,5.571) arc (192.74:347.56:1.503 and 0.771);
      \draw[fill=orange, fill opacity=0.3]
      (4.446,5.574) arc (346.74:553.23:1.511 and 0.734) --
      (1.506,5.574) --
      ++(1.449,-2.711) --
      ++(1.494,2.711) --
      cycle;
      \fill (2.954,2.863) circle (1pt);
      \fill (2.926,3.674) circle (1pt);
      \fill (2.374,6.283) circle (1pt);
    \end{tikzpicture}
    \caption{Two cutting planes}
    \label{fig:conic-two-parabolas}
  \end{subfigure}
  \hfill
  \begin{subfigure}[t]{0.45\textwidth}
    \centering
    \begin{tikzpicture}[yscale=-1]
      \draw (7.334,6.900) -- ++(0.000,-5.206) -- ++(5.206,0.000) -- ++(0.000,5.206) -- cycle;
      \draw[blue] (9.609,6.906) .. controls (8.311,3.260) and (10.306,3.174) .. (12.540,3.837) node[pos=0.3, anchor=east] {$\scriptstyle p_i$};
      \draw[blue] (7.334,5.474) -- (10.934,1.689);
      \draw[red] (7.797,6.903) .. controls ++(3.991,-2.454) and ++(3.214,0.449) .. ++(-0.463,-4.017) node[pos=0.35, anchor=west] {$\scriptstyle p_j$};
      \draw[red] (11.254,6.903) -- (10.100,1.691);
      \draw[densely dashed] (10.454,1.694) -- (9.006,6.900);
      \fill (9.934,4.289) circle (1pt) node[anchor=north] {$\scriptstyle f$};
      \fill (9.903,3.674) circle (1pt) node[anchor=east] {$\scriptstyle q_{ji}$};
      \fill (9.300,5.846) circle (1pt) node[anchor=west] {$\scriptstyle q_{ij}$};
    \end{tikzpicture}
    \caption{Angle bisector and two intersections}
    \label{fig:conic-angle-bisector}
  \end{subfigure}
  \caption{Interpretation of a parabola as a projection of a conic section}
  \label{fig:conic-interpretation}
\end{figure}

\subsection{Two parabolas}
\label{sect:two-parabolas}

We define {\em semi-confocal parabolas} as a family of parabolas that share the same focus,
but their directrixes does not need to have the same directional angle.
It follows that semi-confocal parabolas are projected parabolas from multiple cutting planes that cut the same cone.
We state two important facts on intersections of two parabolas 
which, under the conic interpretation, are straight-forward.

\begin{lemma}
Two semi-confocal parabolas does not intersect iff their directional angles is the same.
\end{lemma}
\begin{proof}
Two parabolas with the same directional angle are produced from two parallel cutting planes in the conic view, which never intersect.
\end{proof}

\begin{lemma}
\label{lemma:angle-bisector}
If two semi-confocal parabolas intersect, then they intersect at two points which also lie on an angle bisector of their directrixes.
\end{lemma}
\begin{proof}
Take conic section interpretation.
Their curves on the surface of the cone must also lie on their cutting planes, which the intersection of the planes is a straight line.
This line piece through the cone exactly two times.
See Figure~\ref{fig:conic-two-parabolas} and \ref{fig:conic-angle-bisector}.
\end{proof}

It follows from Lemma~\ref{lemma:angle-bisector} that two intersecting semi-confocal parabolas $p_i,p_j$ produce a safe region $R = H(p_i) \cap H(p_j)$ which is a bounded convex region.  Moreover, since we only deals with parabolas whose directrixes are from boundary edges of a convex polygon, the non-intersecting case never occurs.

Given two parabola $p_i$ and $p_j$, we can compute their intersections using analytical techniques in $O(1)$ time as follows.
We reduce the problem of finding intersection of two parabolas to the problem of finding intersection of a parabola and a line.
Without loss of generality, we consider $p_i$ in the upright form.
From Lemma~\ref{lemma:angle-bisector}, we find an angle bisector $b$ such that it divide the inner angle between edges $e_i$ and $e_j$.
Then we find the resulting intersections of $p_i$ and $b$ using quadratic equation.

We remark particularly on the structure of the safe region $R$.

Consider each parabola $p_i$ in the upright form.
We can partition this parabola using the two intersection points into 3 arcs: the {\em left arc}, the {\em central arc}, and the {\em right arc}, 
where the left arc corresponds to the half parabola unbounded to $-\infty$ and the right arc corresponds to the half parabola unbounded to $+\infty$, and
the central arc lies between the two intersection points.  (See Figure~\ref{fig:parabola-arc-decomp}.)  
Under this notation, we note that left arc of $p_i$ intersects $p_j$ only once 
at the intersection point where it also intersects the right arc of $p_j$, and vise versa.  

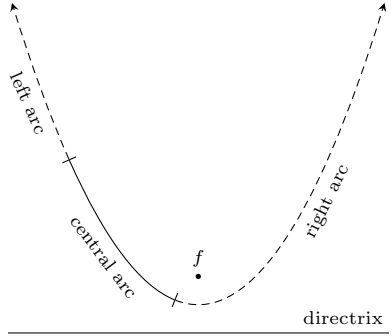
\begin{figure}
  \centering
    \begin{tikzpicture}
      \clip (-2.5,0) rectangle (2.5,5);
      \draw[densely dashed, -stealth] (-1.700,2.552) .. controls (-1.950,3.118) and (-2.200,3.810) .. (-2.450,4.627) node[pos=0.45, anchor=north, rotate=-66] {\scriptsize left arc};
      \draw[densely dashed, -stealth] (-0.300,0.685) .. controls (0.617,0.318) and (1.533,1.632) .. (2.450,4.627) node[pos=0.65, anchor=north, rotate=64] {\scriptsize right arc};
      \draw[|-|] (-1.700,2.552) .. controls (-1.233,1.494) and (-0.767,0.872) .. (-0.300,0.685) node[pos=0.45, anchor=north, rotate=-53] {\scriptsize central arc};
      \fill (0,1) circle (1pt) node[anchor=south] {$\scriptstyle f$};
      \draw (-5,1/4) -- (5,1/4) node[pos=0.69, anchor=south] {\scriptsize directrix};
    \end{tikzpicture}
    \caption{Arc decomposition of a parabola}
    \label{fig:parabola-arc-decomp}
\end{figure}

In our analysis where there are parabolas $p_1,p_2,\ldots,p_n$, we refer to the two intersection points between parabolas $p_i$ and $p_j$ as $q_{ij}$ and $q_{ji}$.  
To distinguish between these two points, imagine one traverses counter-clockwisely on the boundary of $H(p_i)\cap H(p_j)$,
one would see an arc of $p_i$, the intersecting point $q_{ij}$, the arc of $p_j$ and then the intersection point $q_{ji}$.  
The counter-clockwise definitions of these points are crucial to our proof of Lemma~\ref{lemma:qk1-inside}.
See Figure~\ref{fig:conic-angle-bisector}.

\subsection{Many parabolas}

In this section, we analyze the structure of the intersection of $k$ semi-confocal parabolas, extending the result from the previous section.

Let $p_1,p_2,\ldots,p_k$ be $k$ semi-confocal parabolas with different directional angles.
We shall consider the safe region of these parabolas and prove the following lemma.

\begin{lemma}\label{lemma:one-arc}
Each parabola touches at most one arc of the safe region.
\end{lemma}
\begin{proof}
    We prove by induction on $k$.  The case where $k=2$ follows from Lemma~\ref{lemma:angle-bisector}.  Consider $k$ parabolas.  Let $R'=\bigcap_{i=1}^{k-1} H(p_i)$.  Inductively, each parabola $p_1,p_2,\ldots,p_{k-1}$ touches
    at most one arc of $R'$.  We consider $R=R'\cap H(p_k)$, we shall show that $p_k$ only touches at most one arc of $R$.  If $R' \subseteq H(p_k)$, then $R=R'$ and $p_k$ does not touch $R$; hence the lemma is true.
    We then assume that $R' \not\subseteq H(p_k)$.

    Clearly $p_k$ touches one arc of $R=R'\cap H(p_k)$.  Call that arc $a_k$.
    Each endpoint of $a_k$ belongs to some arc of $R'$.

    There are two cases.  
    
    {\em Case 1}: both endpoints of $a_k$ belongs to a single arc of $p_i$.  In this case, $p_k$ intersects exactly one arc of $R'$ exactly twice. Then the safe region $R$ is the intersection of exactly two parabola, i.e. $R=H(p_i)\cap H(p_k)$.  Thus, the lemma follows from Lemma~\ref{lemma:angle-bisector}. 

    {\em Case 2:} one endpoint of $a_k$ belongs to $p_i$ while the other belongs to $p_j$.  We show that in this case, $p_k$ intersects exactly two arcs, implying the lemma.  
    We consider $p_i$ first, i.e., we look at the intersection $R_i=H(p_k)\cap H(p_i)$.   Let $q_i$ be the intersecting point of $p_i$ and $p_k$.  
    We can partition $p_k$ into three arcs; let $b_k$ be the unbounded arc starting at $q_i$.   From Lemma~\ref{lemma:angle-bisector}, we know that $b_k$ only intersect $R_i$ once.   Since $R\subseteq R_i$, we have that $b_k$ also intersects $R$ at most once at $q_i$.  We follow the same argument for $p_j$.  Thus, $p_k$ only intersects $R'$ exactly twice as claimed.    
\end{proof}


\section{The algorithm}
\label{sect:algo}

Our algorithm for finding a safe region works in a similar manner to Graham's scan~\cite{GRAHAM1972132} for convex hull.
We briefly described the algorithm as a pseudocode in Algorithm~\ref{algo:safe-region}.
Later in this section, we explain the algorithm and prove its correctness.

\begin{algorithm}
  \caption{Our algorithm for finding the safe region}
  \label{algo:safe-region}
  \DontPrintSemicolon
  \SetKwBlock{Function}{function}{end function}
  \Function($\textsc{SafeRegion} {(} P : \textsc{Polygon}, f : \textsc{Point} {)}$) {
    $[e_1,e_2,e_3,\ldots,e_n] \gets E(P)$\;
    $p_1 \gets \textsc{Parabola}(f,e_1)$\;
    $p_2 \gets \textsc{Parabola}(f,e_2)$\;
    $R \gets \textsc{DoublyLinkedList}(\textsc{Arc}(p_1,q_{21},q_{12}), \textsc{Arc}(p_2,q_{12},q_{21}))$\;
    \For{$i \in [3,4,5,\dots,n]$} {
      $p_i \gets \textsc{Parabola}(f,e_i)$\;
      $(p_H,\ell_H,r_H) \gets \textsc{Head}(R)$\;
      $(p_T,\ell_T,r_T) \gets \textsc{Tail}(R)$\;
      \If {$q_{iH} \in \textsc{LeftArc}(p_H,\ell_H) \wedge q_{Ti} \in \textsc{RightArc}(p_T,r_T)$} {
        \Continue
        \tcp*{skip this iteration}
      }
      \While {$q_{iH} \in \textsc{RightArc}(p_H,r_H)$} {
        $R \gets \textsc{RemoveHead}(R)$\;
        $(p_H,\underline{\hspace{.5em}},r_H) \gets \textsc{Head}(R)$
        \tcp*{also recompute $q_{iH}$}
      }
      \While {$q_{Ti} \in \textsc{LeftArc}(p_T,\ell_T)$} {
        $R \gets \textsc{RemoveTail}(R)$\;
        $(p_T,\ell_T,\underline{\hspace{.5em}}) \gets \textsc{Tail}(R)$
        \tcp*{also recompute $q_{Ti}$}
      }
      $R \gets \textsc{UpdateHead}(R, \textsc{Arc}(p_H,q_{iH},r_H))$\;
      $R \gets \textsc{UpdateTail}(R, \textsc{Arc}(p_T,\ell_T,q_{Ti}))$\;
      $R \gets \textsc{AppendTail}(R, \textsc{Arc}(p_i,q_{Ti},q_{iH}))$\;
    }
    \Return $R$\;
  }
\end{algorithm}

Since our goal is to find a safe region
\[
R = \bigcap_{i=1}^n H(p_i),
\]
we iterate over parabola $p_i$'s producing a partial solution $R_i$, such that
\[
R_i = \bigcap_{j=1}^i H(p_j),
\]
i.e., $R_i$ is the safe region for the first $i$ parabolas.  We maintain $R_i$ as a cyclic list of parabola arcs
\[
a_1,a_2,\ldots,a_k,
\]
where each arc $a_j$ is a 3-tuple $(p,\ell,r)$ which keeps a reference to the parabola $a_j.p$, its left endpoint $a_j.\ell$, and its right endpoint $a_j.r$.  
We note that with this representation $a_j.r=a_{j+1}.\ell$ for $1\leq j<k$, and $a_k.r=a_1.\ell$.  
We also note that, using the notation defined in Section~\ref{sect:two-parabolas}, $a_j.r$ is $q_{j(j+1)}$ and $a_j.\ell$ is $q_{(j-1)j}$.

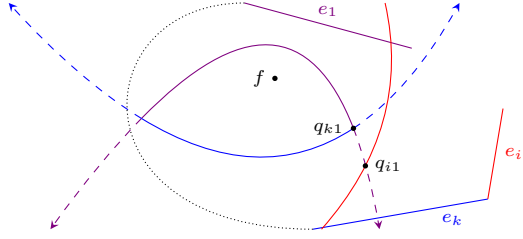
\begin{figure}
  \centering
  \begin{tikzpicture}
    \draw[densely dotted] (1.6,3) .. controls (-0.5,3) and (-0.7,0) .. (2.5,0);
    \draw[violet, dashed, -stealth] (3.033,1.338) .. controls (3.171,0.979) and (3.285,0.533) .. (3.375,0.000);
    \draw[violet, dashed, -stealth] (0.245,1.469) .. controls (-0.124,1.090) and (-0.523,0.600) .. (-0.952,-0.000); 
    \draw[blue, dashed, -stealth] (0.245,1.469) .. controls (-0.308,1.810) and (-0.890,2.320) .. (-1.501,3.000); 
    \draw[blue, dashed, -stealth] (3.033,1.338) .. controls (3.536,1.672) and (4.001,2.226) .. (4.427,3.000);
    \draw[violet] (3.8,2.4) -- (1.6,3) node[pos=0.5, anchor=south] {$\scriptstyle e_1$};
    \draw[blue] (2.5,0) -- (4.8,0.4) node[pos=0.8, anchor=north] {$\scriptstyle e_k$};
    \draw[red] (4.8,0.4) -- (5,1.6) node[pos=0.5, anchor=west] {$\scriptstyle e_i$};
    \draw[violet] (3.033,1.338) .. controls (2.482,2.770) and (1.553,2.813) .. (0.245,1.469);
    \draw[blue] (0.245,1.469) .. controls (1.277,0.833) and (2.206,0.789) .. (3.033,1.338);
    \draw[red] (3.449,3.000) .. controls (3.691,1.913) and (3.413,0.913) .. (2.613,0.000);
    \fill (2,2) circle (1pt) node[anchor=east] {$\scriptstyle f$};
    \fill (3.033,1.338) circle (1pt) node[anchor=east] {$\scriptstyle q_{k1}$};
    \fill (3.19,0.84) circle (1pt) node[anchor=west] {$\scriptstyle q_{i1}$};
  \end{tikzpicture}
  \caption{Proof of Lemma~\ref{lemma:qk1-inside}}
  \label{fig:qk1-inside}
\end{figure}

Initially, we start with $R_2=H(p_1)\cap H(p_2)$.
Which we encode the partial safe region as an ordered list of arcs, $A(R_2) = [a_1,a_2]$, where  
\begin{align*}
a_1 &= (p_1,q_{21},q_{12}) \\
a_2 &= (p_2,q_{12},q_{21})
\end{align*}

For each iteration $i>2$, we consider adding $H(p_i)$ to $R_{i-1}$ to produce $R_i = R_{i-1}\cap H(p_i)$.
There are 3 cases:
\begin{itemize}
\item {\em Case 1:} $p_i$ does not change the region, i.e., $R_{i-1}\cap H(p_i) = R_{i-1}$ and we can discard $p_i$,
\item {\em Case 2:} $p_i$ clips the region, i.e., all parabola arcs in $R_{i-1}$ remains the boundary of $R_{i-1}\cap H(p_i)$, or
\item {\em Case 3:} $p_i$ {\em eclipses} other parabolas in the region, i.e., some parabola arc is entirely outside $R_{i-1}\cap H(p_i)$.
\end{itemize}

To distinguish between these cases, our basic procedure is to test if a point lies in $H(p_i)$.  
The counter-clockwise ordering of parabola ensures that $p_i$ would affect two sequences of arcs, i.e., clockwisely 
\[
a_k,a_{k-1},a_{k-2},\ldots,
\]
to be referred to as the {\em neighbors to the left} of $p_i$, 
and counter-clockwisely,
\[
a_1, a_2, a_3,\ldots,
\]
to be referred to as the {\em neighbors to the right} of $p_i$, 

We first consider point $q_{k1}=a_k.r$ (which is also $a_1.\ell$).  
Lemma~\ref{lemma:qk1-inside} below ensure that we are in case 1 if $q_{k1}\in H(p_i)$.  See Figure~\ref{fig:qk1-inside}.

Otherwise, some part of $R_{i-1}$ is below parabola $p_i$.  We in turns consider points
\[
a_{k-1}.r, a_{k-2}.r, \ldots,
\]
from the neighbors to the left of $p_i$ and find the largest index $j$ such that $a_j \cap H(p_i)\neq \emptyset$.  
In this case, the parabola $a_k.p, a_{k-1}.p,\ldots, a_{j+1}.p$ are eclipsed by $p_i$.

We also process the neighbors to the right of $p_i$ similarly by finding 
the smallest index $j'$ such that $a_{j'}\cap H(p_i)\neq \emptyset$ together with the sequence of eclipsed arcs $a_1,a_2,\ldots,a_{j'-1}$.  We note that it can be the case that $j=j'$, when only one arc survives $p_i$ eclipsing.

To construct $R_i$, we discard eclipsed arcs, add a new arc $a_i$ for $p_i$ and compute 
\begin{itemize}
\item the left intersection point $a_i.\ell=a_j.r$ which is the intersection between $p_i$ and $a_j.p$, and
\item the right intersection point $a_i.r=a_{j'}.\ell$ which is the intersection between $p_i$ and $a_{j'}.p$.
\end{itemize}
Finally, we re-index the arcs in $R_i$.  We quickly remark that this procedure can be seen as a ``twin-headed'' Graham scan.


The following lemmas show that this procedure is correct.

\begin{lemma}
If $q_{k1}\in H(p_i)$, then $R_{i-1}\subseteq H(p_i)$ and $R_i=R_{i-1}$.
\label{lemma:qk1-inside}
\end{lemma}
\begin{proof}
Since $R_{i-1}\subseteq H(p_1)\cap H(p_k)$, our goal is to show that $H(p_1)\cap H(p_k) \subseteq H(p_i)$ in this case.

Consider the intersection of $p_1$ and $p_k$.  Recall that the two intersection points $q_{k1}$ and $q_{1k}$ partition both parabolas into 
their left arcs, central arcs and right arcs.  Let's call them $a^\ell_k,a^c_k,a^r_k$ and  $a^\ell_1,a^c_1,a^r_1$.

We now consider the intersection of $p_1$ and $p_i$.
We show that $q_{i1}$, the intersection point of $p_i$ and $p_1$, is on the left arc $a^\ell_1$ of $p_1$.
To see this, we starts by rotate the plane such that $p_i$ is in the upright form.
Then we find a region $r = H(p_1) \cap H(p_i)$.
Since $q_{k1} \in p_1$ and also $q_{k1}\in H(p_i)$, thus $q_{k1}$ in on the boundary of $H(p_1)\cap H(p_i)=r$.
Again, since $q_{k1}\in p_1$ and also on the boundary of $r$, traversing from $q_{k1}$ on the boundary of $r$ clockwisely w.r.t. $f$
would reach $q_{i1}$, by definition of $q_{i1}$, as claimed.


Using the same argument, we can show that $q_{ki}$ is in the right arc $a^r_k$ of $p_k$.

Using $q_{i1}$ and $q_{ki}$, we partition $p_i$ into 3 arcs: $a_1, a_2,$ and $a_3$ so that $a_1$ is an unbounded curve with $q_{i1}$ as its end point,
$a_2$ is a bounded part with $q_{i1}$ and $q_{ki}$ as their end points, and finally $a_3$ is an unbounded curve with $q_{ki}$ as its end point.  (See Figure~\ref{fig:qk1-inside})

Using the structure from Lemma~\ref{lemma:angle-bisector}, we know that $a_1\cup a_2$ intersects $p_k$ at only $q_{ki}\in a^r_k$.
Thus, $p_i$ does not intersect $a^c_k\cup a^\ell_k$.

Also, we know that $a_2\cup a_3$ intersects $p_1$ at only $q_{i1}\in a^\ell_1$, implying that $p_i$ does not intersect $a^r_1\cup a^c_1$.

We can conclude that $p_i$ does not intersect with $R_{i-1}$, because $R_{i-1}$ lies between the unbounded curves $a^c_k\cup a^\ell_k$ and $a^r_1\cup a^c_1$.
\end{proof}

We also have a simple contraposition.

\begin{corollary}
If $R_{i-1}\not\subseteq H(p_i)$, then $q_{k1}\not\in H(p_i)$.
\label{lemma:not-inside-qk1-out}
\end{corollary}

Let $\hat{R}=\bigcap_{j=1}^i H(p_j)$ be the correct updated solution, we would like to show that $R_i$ constructed above equals $\hat{R}$.  We remark that the $i$-th parabola $p_i$ 
corresponds to edge $e_i$ that comes counter-clockwisely after all other edges that contributes to $R_{i-1}$.

\begin{lemma}
\label{lemma:consecutive-arcs}
If $R_{i-1}\not\subseteq H(p_i)$, the arcs on $R_{i-1}$'s boundary which do not belong to $\hat{R}$ form a consecutive sequences
\[
a_j,a_{j+1},\ldots,a_k,a_1,a_2,\ldots,a_{j'}.
\]
\end{lemma}
\begin{proof}
Lemma~\ref{lemma:one-arc} ensures that if $p_i$ intersects the boundary of $R_{i-1}$, $p_i$ touches at most one arc of $\hat{R}$, the resulting of the intersection of $H(p_i)$ and $R_{i-1}$.  This implies that the arcs of $R_{i-1}$ not belonging to forms a (circular) consecutive sequence.  To see this, assume otherwise and note that in that case $p_i$ would touches more than one arcs of $\hat{R}$.

Our procedure finds the consecutive sequence starting at $q_{k1}$, the intersection of $p_k$ and $p_1$, then iterates through other consecutive points.  Thus, the procedure is correct if the starting point is correct, i.e., we start at some intersection point outside $\hat{R}$.  This is indeed the case because Corollary~\ref{lemma:not-inside-qk1-out} guarantees that when $R_{i-1}\not\subseteq H(p_i)$, $q_{k1}$ is outside $\hat{R}$.
\end{proof}

We conclude with our main correctness theorem.

\begin{theorem}
Our updating procedure is correct, i.e.,
$R_{i+1}=\hat{R}$, and the algorithm computes the safe region in linear time.
\end{theorem}
\begin{proof}
Regarding the updating procedure, we deals with 3 possible cases.
Lemma~\ref{lemma:qk1-inside} ensures that our condition for case 1 is correct.  In other cases, Lemma~\ref{lemma:consecutive-arcs} shows that the procedure for deleting arcs is correct.  By induction on $n$, the algorithm is thus produces the required safe region.

To analyze the running time, we first note that, except the two inner while loops, for each $i$, the algorithm runs in $O(1)$ time.  To account for the running time of the inner while loops, observe that each iteration of the loop removes one parabola from the list.  Since at most $n$ parabolas are inserted in the list, the deletion can take place at most $n$ times, implying the total running time of $O(n)$ for the loops.
\end{proof}

\newcommand{\interior}{\mathrm{int}}

\section{The number of arcs of the safe region}
\label{sect:skeleton}

In this section, we consider the complexity of the boundary of the safe region;
in other words, we count the exact number of arcs of the safe region.
From previous sections, we derive that a side of the safe region is a parabolic arc with point $f$ as its focus.
We also see that some edge of $P$ may not contribute to the resulting safe region, i.e., a parabola associated with it does not touch the safe region.  It is natural to ask for the number of arcs of the safe region.

Assuming that the polygon $P$ is fixed, the number of arcs of the safe region depends on the focus $f$.
We denote explicitly by $R_f$ a safe region with point $f$ as its focus.  
As in~\cite{AkitayaBDH021}, this section analyzes the number of arcs of safe region $R_f$, i.e., $|A(R_f)|$,
where $A(R_f)$ is the set of arcs of $R_f$.
Figure~\ref{fig:nos-arcs-triangle}(a) shows two safe regions $R_{f_2}$ with query point $f_2$ and $R_{f_3}$ with query point $f_3$.

Akitaya {\em et. al.}~\cite{AkitayaBDH021} consider the same problem for the case where each side of $P$ is folded onto a line.  
They show that straight skeleton of $P$ plays an important role in determining the number of sides of the resulting region.  
It is true for our case as well.  
As an example, Figure~\ref{fig:nos-arcs-triangle}(b) shows an inscribed circle $C$ which can be determined using straight skeleton
and two safe regions shown previously.  We remark that $f_2\not\in C$ but $f_3\in C$.

We start by defining useful notations related to straight skeleton and event circles.
A {\em straight skeleton}~\cite{AicholzerA-cocoon96-skeletons} of a polygon $P$, denoted by $S(P)$, is a subset of $P$ such that for each point $u \in S(P)$, there exist at least two points on $\delta P$ with the same distance to $u$.
More intuitively, we may see the skeleton as a Voronoi diagram of line segments where each site is an edge of the polygon.
The straight skeleton $S(P)$ partitions $P$ into regions, referred to as {\em faces}.  
Thus, under the Voronoi interpretation, each face is bounded by exactly one polygon edge as other edges of $S(P)$.
We call a face that is bounded by polygon edge $e_i$ as face $F_i$.  
We also note that a face is also a convex polygon.  
See Figure~\ref{fig:skeleton} for an illustration.

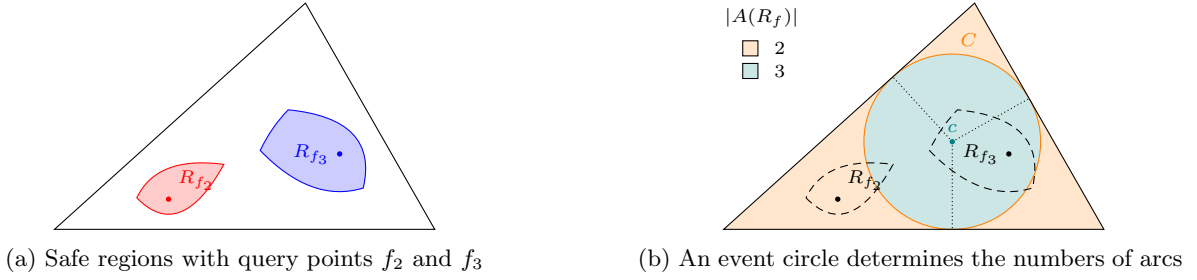
\begin{figure}
  \centering
  \begin{subfigure}[t]{0.45\textwidth}
    \centering
    \begin{tikzpicture}
      \draw (0.000,0.000) -- (5.000,0.000) -- (3.300,3.000) -- cycle;
      \filldraw[blue!20, draw=blue]
        (2.705,1.046) .. controls (3.157,0.574) and (3.608,0.408) ..
        (4.060,0.548) .. controls (4.201,1.160) and (3.873,1.506) ..
        (3.074,1.585) .. controls (2.925,1.433) and (2.803,1.254) .. cycle;
      \fill[blue] (3.750,1.000) circle (1pt) node[anchor=east] {$\scriptstyle R_{f_3}$};
      \filldraw[red!20, draw=red]
        (1.082,0.418) .. controls (1.464,0.020) and (1.846,0.167) ..
        (2.227,0.861) .. controls (1.639,0.940) and (1.258,0.793) .. cycle;
      \fill[red] (1.500,0.400) circle (1pt) node[anchor=south west] {$\scriptstyle R_{f_2}$};
    \end{tikzpicture}
    \caption{Safe regions with query points $f_2$ and $f_3$}
  \end{subfigure}
  \hfill
  \begin{subfigure}[t]{0.45\textwidth}
    \centering
    \begin{tikzpicture}
      \fill[orange!20] (0.000,0.000) -- (5.000,0.000) -- (3.300,3.000) -- cycle;
      \draw[orange, fill=teal!20] (3.006,0) arc (-90:270:1.162) node[pos=0.47, anchor=south] {$\scriptstyle C$};
      \draw[densely dotted] (3.006,1.162) -- (3.006,0.000);
      \draw[densely dotted] (3.006,1.162) -- (2.224,2.022);
      \draw[densely dotted] (3.006,1.162) -- (4.017,1.735);
      \fill[teal] (3.006,1.162) circle (1pt) node[anchor=south] {$\scriptstyle c$};
      \draw (0.000,0.000) -- (5.000,0.000) -- (3.300,3.000) -- cycle;
      \draw[densely dashed]
        (2.705,1.046) .. controls (3.157,0.574) and (3.608,0.408) ..
        (4.060,0.548) .. controls (4.201,1.160) and (3.873,1.506) ..
        (3.074,1.585) .. controls (2.925,1.433) and (2.803,1.254) .. cycle;
      \fill (3.750,1.000) circle (1pt) node[anchor=east] {$\scriptstyle R_{f_3}$};
      \draw[densely dashed]
        (1.082,0.418) .. controls (1.464,0.020) and (1.846,0.167) ..
        (2.227,0.861) .. controls (1.639,0.940) and (1.258,0.793) .. cycle;
      \fill (1.500,0.400) circle (1pt) node[anchor=south west] {$\scriptstyle R_{f_2}$};
      \draw (0.5,2.8) node {$\scriptstyle |A(R_f)|$};
      \draw[fill=orange!20] (0.25,2.3) rectangle ++(0.2,0.2) node[pos=0.5, anchor=west] {$\;\;\scriptstyle 2$};
      \draw[fill=teal!20] (0.25,2.0) rectangle ++(0.2,0.2) node[pos=0.5, anchor=west] {$\;\;\scriptstyle 3$};
    \end{tikzpicture}
    \caption{An event circle determines the numbers of arcs}
  \end{subfigure}
  \caption{Safe regions with 2 and 3 parabola arcs}
  \label{fig:nos-arcs-triangle}
\end{figure}

The skeleton may be viewed as a tree, where each non-leaf node ensures at least three equidistant points on $\delta P$.
A non-leaf node of $S(P)$ is referred to as an {\em event point}.
A circle centered at an event point and tangent to the nearest edge of the polygon is called an {\em event circle} of $S(P)$.
Let $\mathcal{C}$ be the set of all event circles of $S(P)$, 
and let $\mathcal{C}_e \subseteq \mathcal{C}$ be a set of event circles tangent to edge $e$.
We also denote by $\interior(C)$ an interior of circle $C$, i.e., the set $\{ (x,y) : (x-x_0)^2+(y-y_0)^2<r^2\}$ for a circle with center $(x_0,y_0)$ and radius $r$.  

The goal of this section is to show that, under the fixed polygon $P$, 
the structure of $|A(R_f)|$ is governed by event circles of straight skeleton $S(P)$ of $P$.

We start by analyzing the case when the safe region intersects with the skeleton faces.
The following lemma directly follows from the Voronoi interpretation of the straight skeleton.

\begin{lemma}
\label{lemma:face-tangent-edge}
For each event circle $C$ with event point $c$, $c$ is adjacent to face $F_i$, 
if and only if its corresponding polygon edge $e_i$ is tangent to $C$.
\label{lemma:adjacent-face-tangent-edge}
\end{lemma}

The following two lemmas provide basic properties for our main theorem in this section.

\begin{lemma}
For a particular focus $f$, $\interior(R_f)$ intersects with skeleton face $F_i$ adjacent to edge $e_i$, 
iff the associated parabola $p_i$ is part of the arcs of $R_f$.
\label{lemma:arc-appears-if-intersect-face}
\end{lemma}
\begin{proof}

$(\Rightarrow)$ Assume that $\interior(R_f)$ intersects $F_i$.  
Since $F_i$ is bounded by a polygon edge and $R_f$ is contained in $P$, 
we know that there are parts of the boundary of $R_f$ that intersect $F_i$.
Consider any point $u$ on the boundary of $R_f$ inside $F_i$.
Clearly $u$ must be on an arc of some parabola, i.e.,
we have that 
\[
\min_j |\overline{uu_j}| = |\overline{uf}|,
\]
where $u_j$ is an orthogonal projection of $u$ onto $e_j$.
Since all points in $F_i$ is closer to $e_i$ than other edges,
the minimizer of the above term is $u_i$; thus,
$u$ must also lies on $p_i$, i.e., $p_i$ is part of arcs of $R_f$.

$(\Leftarrow)$  We prove by contradiction.  
Assume that $\interior(R_f)$ does not intersect $F_i$, but $p_i$ is part of the boundary of $R_f$.
Consider any point $u\in p_i$ on the boundary.  
Since $\interior(R_f)$ is disjoint from $F_i$, $u$ is strictly in some face $F_j$.
In this case we have that
\[
|\overline{uf}| = |\overline{uu_i}| > |\overline{uu_j}|,
\]
where $u_i$ and $u_j$ are orthogonal projections of $u$ onto $e_i$ and $e_j$,
implying that $u\not\in H(p_j)$, a contradiction.
\end{proof}

\begin{lemma}
\label{lemma:crossing-point-generalized}
For $C \in \mathcal{C}$ with event point $c$,
a safe region $R_f$ strictly contains an event point $c$, i.e., $c\in \interior(R_f)$,
iff $f \in \interior(C)$.
\end{lemma}
\begin{proof}
Let $r$ be radius of event circle $C$.
Project $c$ orthogonally to a line extension of every edge $e_i$, named the projected point $c_i$.

$(\Leftarrow)$ Assume that $f \in \interior(C)$, i.e., $|\overline{cf}|<r$.  We show that $r\in H(p_i)$ for every parabola $p_i$ associated with polygon edge $e_i$.
This is the case when $r$ is strictly closer to $f$ than every other edge $e_i$.
Consider each edge $e_i$ tangent to $C$, we have $|\overline{cf}| < r = |\overline{cc_i}|$.
For edge $e_i$ not tangent to $C$, we have that $|\overline{cc_i}|>r$; thus, $|\overline{cf}| < r < |\overline{cc_i}|$.
Hence, $c$ is in the safe region $R_f$.

$(\Rightarrow)$ Assume that $f \not\in \interior(C)$.  In this case, we have $|\overline{cf}|\geq r$.
Since $C$ is an event circle, there exists edge $e_i$ tangent to $C$.  For that particular edge, we have $|\overline{cc_i}|=r$.
Thus, $|\overline{cf}| \ge r = |\overline{cc_i}|$, and $c$ is not strictly contained in the safe region $R_f$.
\end{proof}

The following theorem gives the number of boundary arcs of $R_f$ as a function of event circles containing $f$.

\begin{theorem}
\label{theorem:count-arcs}
If $f$ is strictly inside some event circle, i.e., $f\in \interior(C)$ for some $C\in{\mathcal{C}}$, then 
\[
|A(R_f)|=
\left|\{ e_i \in E(P) : \mbox{there exists $C \in \mathcal{C}_{e_i}$ s.t. $f \in \interior(C)$}\}\right|
\]
Otherwise, $|A(R_f)|=2$.
\end{theorem}
\begin{proof}
We first assume that $f\in\interior(C)$ for some event circle $C$.
Consider each event circle $C$ with event point $c$ such that $f\in\interior(C)$.
From Lemma~\ref{lemma:crossing-point-generalized}, we know that
$c\in \interior(R_f)$. This also means that $\interior(R_f)$ intersects every face $F_i$ adjacent to event point $c$.
Lemma~\ref{lemma:adjacent-face-tangent-edge} ensures that these faces $F_i$'s correspond with edges $e_i$'s tangent to $C$, the set of edges $e_i$ such that $C\in \mathcal{C}_{e_i}$.
Since Lemma~\ref{lemma:arc-appears-if-intersect-face} ensures that for each face $F_i$ adjacent to edge $e_i$ intersecting with $R_f$, the parabola $p_i$ appears as an arc of $R_f$, we have that for each tangent edge $e_i$ of $C$, its parabola $p_i$ appears as an arc in $R_f$.
The lemma, in this case, follows by taking the union of all boundary edges from every event circle $C\in\mathcal{C}$ that $f$ is strictly inside.

On the other hand, if $f$ is not strictly contained in any event circle $C \in \mathcal{C}$,
Lemma~\ref{lemma:angle-bisector} ensures that the safe region must touches two parabolas.
\end{proof}

Figure~\ref{fig:circles-count-arcs} shows an application of Theorem~\ref{theorem:count-arcs}.

Alternatively, one may view Theorem~\ref{theorem:count-arcs} with the conic section interpretation as follows.
The input polygon $P$ induces $n$ cutting planes, forming the straight skeleton and their corresponding faces when projected onto $xy$-plane, and the point $f$ is represented as a cone whose apex is at $f$.

The structural results in this section give another linear-time algorithm for finding safe regions, by first finding straight skeleton in $O(n)$-time using~\cite{ChinSW99-medial-axis}, then computing event circles, and finally using this information to find the set of edges contributing to the arcs of the safe region.  
However, we believe that the results in this section contribute mainly to the structural understanding of the problem and may serve as a guideline for tackling harder problems, especially the non-convex case of the problem.  
We discuss this in Section~\ref{sect:non-convex-future}.

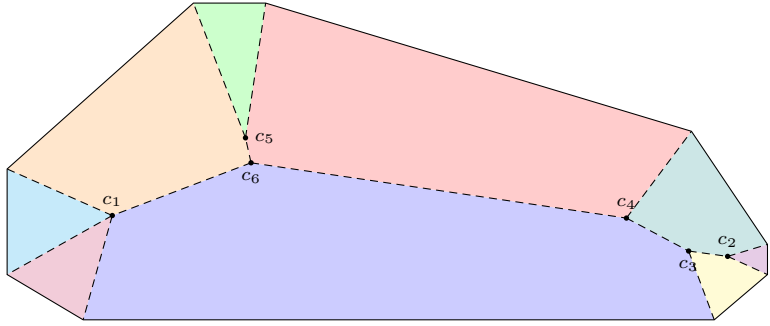
\begin{figure}
  \centering
  \begin{tikzpicture}
    \coordinate (v1) at (0,2);
    \coordinate (v2) at (0,0.6);
    \coordinate (v3) at (1,0);
    \coordinate (v4) at (9.3,0);
    \coordinate (v5) at (10,0.6);
    \coordinate (v6) at (10,1);
    \coordinate (v7) at (9,2.5);
    \coordinate (v8) at (3.4,4.2);
    \coordinate (v9) at (2.45,4.2);
    \coordinate (c1) at (1.383,1.383);
    \coordinate (c2) at (9.476,0.841);
    \coordinate (c3) at (8.963,0.912);
    \coordinate (c4) at (8.145,1.349);
    \coordinate (c5) at (3.135,2.413);
    \coordinate (c6) at (3.208,2.082);
    \fill[cyan!20]   (v1) -- (v2) -- (c1) -- cycle;
    \fill[purple!20] (v2) -- (v3) -- (c1) -- cycle;
    \fill[blue!20]   (v3) -- (v4) -- (c3) -- (c4) -- (c6) -- (c1) -- cycle;
    \fill[yellow!20] (v4) -- (v5) -- (c2) -- (c3) -- cycle;
    \fill[violet!20] (v5) -- (v6) -- (c2) -- cycle;
    \fill[teal!20]   (v6) -- (v7) -- (c4) -- (c3) -- (c2) -- cycle;
    \fill[red!20]    (v7) -- (v8) -- (c5) -- (c6) -- (c4) -- cycle;
    \fill[green!20]  (v8) -- (v9) -- (c5) -- cycle;
    \fill[orange!20] (v9) -- (v1) -- (c1) -- (c6) -- (c5) -- cycle;
    \draw (v1) -- (v2) -- (v3) -- (v4) -- (v5) -- (v6) -- (v7) -- (v8) -- (v9) -- cycle;
    \draw[densely dashed] (c1) -- (v1);
    \draw[densely dashed] (c1) -- (v2);
    \draw[densely dashed] (c1) -- (v3);
    \draw[densely dashed] (c2) -- (v5);
    \draw[densely dashed] (c2) -- (v6);
    \draw[densely dashed] (c3) -- (c2);
    \draw[densely dashed] (c3) -- (v4);
    \draw[densely dashed] (c4) -- (c3);
    \draw[densely dashed] (c4) -- (v7);
    \draw[densely dashed] (c5) -- (v8);
    \draw[densely dashed] (c5) -- (v9);
    \draw[densely dashed] (c6) -- (c1);
    \draw[densely dashed] (c6) -- (c4);
    \draw[densely dashed] (c6) -- (c5);
    \fill (c1) circle (1pt) node[anchor=south] {$\scriptstyle c_1$};
    \fill (c2) circle (1pt) node[anchor=south] {$\scriptstyle c_2$};
    \fill (c3) circle (1pt) node[anchor=north] {$\scriptstyle c_3$};
    \fill (c4) circle (1pt) node[anchor=south] {$\scriptstyle c_4$};
    \fill (c5) circle (1pt) node[anchor=west]  {$\scriptstyle c_5$};
    \fill (c6) circle (1pt) node[anchor=north] {$\scriptstyle c_6$};
  \end{tikzpicture}
  \caption{Skeleton of polygon $P$, each face colored differently}
  \label{fig:skeleton}
\end{figure}

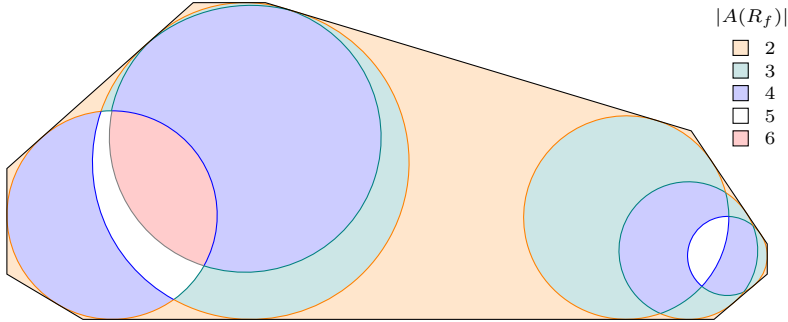
\begin{figure}
  \centering
  \begin{tikzpicture}
    \coordinate (v1) at (0,2);
    \coordinate (v2) at (0,0.6);
    \coordinate (v3) at (1,0);
    \coordinate (v4) at (9.3,0);
    \coordinate (v5) at (10,0.6);
    \coordinate (v6) at (10,1);
    \coordinate (v7) at (9,2.5);
    \coordinate (v8) at (3.4,4.2);
    \coordinate (v9) at (2.45,4.2);
    \coordinate (c1) at (1.383,1.383);
    \coordinate (c2) at (9.476,0.841);
    \coordinate (c3) at (8.963,0.912);
    \coordinate (c4) at (8.145,1.349);
    \coordinate (c5) at (3.135,2.413);
    \coordinate (c6) at (3.208,2.082);
    \draw (9.8,4.0) node {$\scriptstyle |A(R_f)|$};
    \draw[fill=orange!20] (9.55,3.5) rectangle ++(0.2,0.2) node[pos=0.5, anchor=west] {$\;\;\scriptstyle 2$};
    \draw[fill=teal!20]   (9.55,3.2) rectangle ++(0.2,0.2) node[pos=0.5, anchor=west] {$\;\;\scriptstyle 3$};
    \draw[fill=blue!20]   (9.55,2.9) rectangle ++(0.2,0.2) node[pos=0.5, anchor=west] {$\;\;\scriptstyle 4$};
    \draw[fill=white!20]  (9.55,2.6) rectangle ++(0.2,0.2) node[pos=0.5, anchor=west] {$\;\;\scriptstyle 5$};
    \draw[fill=red!20]    (9.55,2.3) rectangle ++(0.2,0.2) node[pos=0.5, anchor=west] {$\;\;\scriptstyle 6$};
    \fill[orange!20] (v1) -- (v2) -- (v3) -- (v4) -- (v5) -- (v6) -- (v7) -- (v8) -- (v9) -- cycle;
    \fill[blue!20]   (c1) circle (1.383);
    \fill[teal!20]   (c5) circle (1.787);
    \fill[teal!20]   (c6) circle (2.082);
    \fill[blue!20]   (3.730,4.098) arc (75.5:129.5:2.082) arc (-225.5:70.5:1.787) -- cycle;
    \fill[white!20]  (1.239,2.759) arc (161.0:240.8:2.082) arc (-54.1:96.0:1.383) -- cycle;
    \fill[red!20]    (1.383,2.766) arc (168.6:252.3:1.787) arc (-29.1:90.0:1.383) -- cycle;
    \fill[teal!20]   (c2) circle (0.524);
    \fill[teal!20]   (c3) circle (0.920);
    \fill[teal!20]   (c4) circle (1.349);
    \fill[blue!20]   (9.81,1.24) arc (50.1:294.2:0.524) arc (-37.0:21.3:0.91) -- cycle;
    \fill[blue!20]   (9.45,1.68) arc (57.4:246.3:0.91) arc (-70.5:14.2:1.35) -- cycle;
    \fill[white!20]  (9.50,1.37) arc (87.9:230.3:0.524) arc (-42.5:0.7:1.35) -- cycle;
    \draw[teal]   (1.383,2.766) arc (90.0:96.0:1.383);
    \draw[orange] (1.239,2.759) arc (96.0:360-54.1:1.383);
    \draw[teal]   (2.195,0.263) arc (-54.1:-29.1:1.383);
    \draw[blue]   (2.591,0.710) arc (-29.1:90.0:1.383);
    \draw[orange] (9.691,0.363) arc (294.2:360+50.1:0.524);
    \draw[teal]   (9.812,1.243) arc (50.1:87.9:0.524);
    \draw[blue]   (9.495,1.365) arc (87.9:230.2:0.524);
    \draw[teal]   (9.140,0.438) arc (230.2:294.2:0.524);
    \draw[orange] (9.812,1.243) arc (21.3:57.4:0.912);
    \draw[teal]   (9.453,1.681) arc (57.4:246.2:0.912);
    \draw[orange] (8.595,0.077) arc (246.2:323.0:0.912);
    \draw[teal]   (9.691,0.363) arc (323.0:360+21.3:0.912);
    \draw[orange] (9.453,1.681) arc (14.2:289.5:1.349);
    \draw[teal]   (8.595,0.077) arc (289.5:317.5:1.349);
    \draw[blue]   (9.140,0.438) arc (317.5:360+0.7:1.349);
    \draw[teal]   (9.495,1.365) arc (0.7:14.2:1.349);
    \draw[orange] (3.730,4.098) arc (70.5:134.5:1.787);
    \draw[teal]   (1.882,3.688) arc (134.1:168.6:1.787);
    \draw[gray]   (1.383,2.766) arc (168.6:252.3:1.787);
    \draw[teal]   (2.591,0.710) arc (252.3:360+70.5:1.787);
    \draw[teal]   (3.730,4.098) arc (75.5:129.5:2.082);
    \draw[orange] (1.882,3.688) arc (129.5:161.0:2.082);
    \draw[blue]   (1.239,2.759) arc (161.0:240.9:2.082);
    \draw[orange] (2.195,0.263) arc (240.9:360+75.5:2.082);
    \draw (v1) -- (v2) -- (v3) -- (v4) -- (v5) -- (v6) -- (v7) -- (v8) -- (v9) -- cycle;
  \end{tikzpicture}
  \caption{Event circles from straight skeleton determine $|A(R_f)|$ (same polygon as Figure~\ref{fig:skeleton})}
  \label{fig:circles-count-arcs}
\end{figure}

\section{Conclusions and open problems}

We give a linear-time algorithm for finding a safe region for folding each point on the boundary of a convex polygon $P$ to point $f\in P$.  
We also give structural properties related to the number of arcs in a safe region for each focal point $f$, based on straight skeletons.
We note the crucial roles straight skeletons in our problem as well as other problems in origami design, as can be seen in~\cite{lang1996computational}.  
An interesting direction for future work is to investigate problems with the similar structures while using straight skeletons as keys.

As mentioned in the introduction, we also hope that our results show interesting connections between the two problems posted by Akitaya {\em et al.}~\cite{AkitayaBDH021}.

\subsection{Remarks on non-convex polygons}
\label{sect:non-convex-future}
Results in Section~\ref{sect:skeleton} shed some light to non-convex cases.  
However, there are issues with the current approach.  
When dealing with non-convex polygons, there are two related concepts: straight skeletons~\cite{AicholzerA-cocoon96-skeletons} and medial axes~\cite{Blum_1967_6755} (see also~\cite{ChinSW99-medial-axis, Attali2009}).  
For a given polygon, a medial axis contains points with equal distance to more than one points on $\delta P$,
while a straight skeleton is a Voronoi diagram where each site is a {\em line extension} of each edge. 
They are the same in convex polygons, however, in non-convex polygons, their medial axes contains curved segments.

In our application, since we can only fold every point on each polygon edge, but not points on the line extension of the edge, it make sense to consider a medial axis.  
However, each face in the medial axis can be bounded by more than one polygon edges, breaking down one of our key assumptions.
We leave the investigation of this approach to non-convex polygons as an important open question.

\section{Acknowledgements}

The authors would like to thank organizers and participants of JCDCGGG
2022, especially Hugo Akitaya, for giving helpful feedbacks on the
presentation of this work and pointing out important issues related to
the problem.

\bibliography{folding}
\bibliographystyle{plain}

\end{document}